\documentclass[runningheads,a4paper]{llncs}

\usepackage{latexsym}
\usepackage{amssymb}
\usepackage{algorithm}
\usepackage{algorithmic}
\usepackage{graphicx}
\usepackage{enumerate}
\usepackage{url}
\urldef{\mailsa}\path|{alfred.hofmann, ursula.barth, ingrid.haas, frank.holzwarth,|
\urldef{\mailsb}\path|anna.kramer, leonie.kunz, christine.reiss, nicole.sator,|
\urldef{\mailsc}\path|erika.siebert-cole, peter.strasser, lncs}@springer.com|
\newcommand{\keywords}[1]{\par\addvspace\baselineskip
\noindent\keywordname\enspace\ignorespaces#1}

\newcommand{\ZZ}{{\mathbb Z}}

\floatname{algorithm}{Protocol}

\newcommand{\bb}{\mbox{\bf b}}

\newcommand{\ba}{\mbox{\bf a}}

\newcommand{\bs}{\mbox{\bf s}}

\newcommand{\bTheta}{\mbox{$\bf \Theta$}}

\newcommand{\be}{\begin{equation}}
\newcommand{\ee}{\end{equation}}

\newcommand{\mod}{\mbox{~mod~}}

\newcommand{\comment}[1]{}

\begin{document}

\title{Secure Mining of Association Rules in Horizontally Distributed Databases}

\titlerunning{Mining of Association Rules in Horizontally Distributed Databases}

\author{Tamir Tassa\inst{1}}
\authorrunning{T. Tassa} % abbreviated author list (for running head)
%Craig Chambers, Kim B. Bruce, and Elisa Bertino}
%
\institute{Department of Mathematics and Computer Science, The Open
University, Israel}

%\toctitle{Anonymizing Tables without Clustering}

\maketitle

\begin{abstract}
We propose a protocol for secure mining of association rules in horizontally distributed databases.
The current leading protocol is that of Kantarcioglu and Clifton \cite{KanCli}. Our protocol, like theirs,
is based on the Fast Distributed Mining (FDM) algorithm of Cheung et al. \cite{article2}, which is
an unsecured distributed version of the Apriori algorithm. The main ingredients in our protocol are two novel secure multi-party algorithms ---
one that computes the union of private subsets that each of the interacting players hold, and another that tests the inclusion of an
element held by one player in a subset held by another.
Our protocol offers enhanced privacy with respect to the protocol in \cite{KanCli}. In addition, it is simpler and
is significantly more efficient in terms of
communication rounds, communication cost and computational cost.

\end{abstract}

\section{Introduction}
We study here the problem of secure mining of association rules in horizontally partitioned databases.
In that setting, there are several sites (or players) that hold homogeneous databases, i.e., databases that share the same schema but hold
information on different entities. The goal is to find all association rules with given minimal support and confidence levels
that hold in the unified database, while minimizing the
information disclosed about the private databases held by those players.
%The information that we would like to protect in this context is not only
%individual transactions in the different databases, but also more global information such as what association rules are supported locally in each of those
%databases.

%\vskip -0.01cm
That goal defines a problem of secure multi-party computation. In such problems,
there are $M$ players that hold private inputs, $x_1,\ldots,x_M$, and they wish to securely compute
$y=f(x_1,\ldots,x_M)$ for some public function $f$. If there existed a trusted third party, the players could surrender to him their inputs and he would
perform the function evaluation and send to them the resulting output. In the absence of such a trusted third party, it is needed to devise
a protocol that the players can run on their own in order to arrive at the required output $y$. Such a protocol is considered
perfectly secure if no player can learn from his view of the protocol more than what he would have learnt in the idealized setting where the computation is carried out
by a trusted third party. Yao \cite{Yao} was the first to propose a generic solution for this problem in the case of two players.
Other generic solutions, for the multi-party case, were later proposed in \cite{BMR,BNP,GMW87}.

In our problem, the inputs are the partial databases, and the required output is the list of association rules with given support and confidence.
As the above mentioned generic solutions rely upon a description of the function $f$ as a Boolean circuit, they can be applied only to small inputs and functions which are
realizable
by simple circuits. In more complex settings, such as ours,
other methods are required for carrying out this computation.
In such cases, some relaxations of the notion of perfect security might be inevitable when looking for practical protocols,
provided that the excess information is deemed benign (see examples of such protocols in e.g. \cite{KanCli,VC02,ZYW}).

%\vskip -0.01cm
Kantarcioglu and Clifton studied that problem in \cite{KanCli} and devised a protocol for its solution.
The main part of the protocol is a sub-protocol for the secure computation of the union of private subsets that are held by the different players. (Those subsets
include candidate itemsets, as we explain below.) That is the most costly part of the protocol and its implementation
relies upon cryptographic primitives such as commutative encryption, oblivious transfer, and hash functions.
This is also the only part in the protocol in which the players may extract from their view of the protocol
information on other databases, beyond what is implied by the final output and their
own input. While such leakage of information renders the protocol not perfectly secure, the perimeter of the excess information is explicitly bounded in \cite{KanCli}
and it is argued that such information leakage is innocuous, whence acceptable
from practical point of view.

%\vskip -0.01cm
Herein we propose an alternative protocol for the secure computation of the union of private subsets.
The proposed protocol improves upon that in \cite{KanCli} in terms
of simplicity and efficiency as well as privacy. In particular, our protocol
does not depend on commutative encryption and oblivious transfer (what simplifies
it significantly and contributes towards reduced communication and computational costs). While our solution is still not perfectly secure, it
leaks excess information only to a small number of coalitions (three),
unlike the protocol of \cite{KanCli} that discloses information also to some single players. In addition, we claim that the excess information
that our protocol may leak is less sensitive than the
excess information leaked by the protocol of \cite{KanCli}.

%\vskip -0.01cm
The protocol that we propose here computes a parameterized family of functions, which we call threshold functions, in which the two extreme cases
correspond to the problems of computing the union and intersection of private subsets. Those are in fact general-purpose protocols that can be used
in other contexts as well. Another problem of secure multi-party computation that we solve here as part of our discussion is the problem of determining whether
an element held by one player is included in a subset held by another.

\subsection{Preliminaries}\label{prelim}
Let $D$ be a transaction database.
As in \cite{KanCli}, we view $D$ as a binary matrix of $N$ rows and $L$ columns, where each row is a transaction over some set of items, $A=\{a_1,\ldots,a_L\}$,
and each column represents one of the items in $A$. (In other words, the $(i,j)$th entry of $D$ equals 1 if the $i$th transaction includs the item $a_j$, and 0 otherwise.)
The database $D$ is partitioned horizontally between $M$ players, denoted $P_1,\ldots,P_{M}$.
Player $P_m$ holds the partial database $D_m$
that contains $N_m=|D_m|$ of the transactions in $D$, $1 \leq m \leq M$. The unified database is $ D = D_1\cup  \cdots \cup D_{M} $, and $N = \sum_{m=1}^M N_m$.

An itemset $X$ is a subset of $A$. Its global support, $supp(X)$, is the number of transactions in $D$ that contain it. Its local support, $supp_m(X)$, is the number of
transactions in $D_m$ that contain it. Clearly, $supp(X)=\sum_{m=1}^{M} supp_m(X)$.
Let $s$ be a real number between 0 and 1 that stands for a required threshold support. An itemset $X$ is called $s$-frequent if $supp(X) \geq s N$.
It is called locally $s$-frequent
at $D_m$ if $supp_m(X) \geq s N_m$.

For each $1 \leq k \leq L$, let $F_s^k$ denote the set of all $k$-itemsets (namely, itemsets of size $k$) that are $s$-frequent,
and $F_s^{k,m}$ be the set of all $k$-itemsets that are locally $s$-frequent at $D_m$, $1 \leq m \leq M$.
Our main computational goal is to find, for a given threshold support $0 < s \leq 1$, the set of all $s$-frequent itemsets, $F_s:=\bigcup_{k=1}^L F_s^k$.
We may then continue to find all $(s,c)$-association rules, i.e., all association rules of support at least $sN$ and confidence at least $c$.
(Recall that if $X$ and $Y$ are two disjoint subsets of $A$, the support of the corresponding association rule $X \Rightarrow Y$ is $supp(X \cup Y)$ and its confidence
is $supp(X \cup Y)/supp(X)$.)

\comment{@@@
Assume that we wish to mine all association rules with support at least $sN$ and confidence at least $c$ (we refer to such rules as $(s,c)$-rules).
If $Z \in F_s$ is one of the $s$-frequent itemsets and $Z = X \cup Y$ is a decomposition of $Z$ into two disjoint subsets, then $X \Rightarrow Y$ is an $(s,c)$-rule
if and only if $supp(X \cup Y)/supp(X) \geq c$.
If in the first stage, the players found all $s$-frequent itemsets as well as their support counts,
then that information suffices for verifying the $c$-confidence of the candidate rules.
If, however,
the output of the first stage includes just the list of $s$-frequent itemsets, without their support counts, the verification of $c$-confidence
entails
another secure multi-party protocol.
}

\subsection{The Fast Distributed Mining algorithm}
The protocol of \cite{KanCli}, as well as ours, are based on the
Fast Distributed Mining (FDM) algorithm of Cheung et al. \cite{article2}, which is
an unsecured distributed version of the Apriori algorithm.
Its main idea is that any $s$-frequent itemset must be
also locally $s$-frequent in at least one of the sites. Hence, in order to find
all globally $s$-frequent itemsets, each player reveals his locally $s$-frequent itemsets and then the players
check each of them to see if they are $s$-frequent also globally.
The stages of the FDM algorithm are as follows:

\begin{enumerate}[(1)]
    \item {\bf Initialization:} It is assumed that the players have already jointly calculated $F_s^{k-1}$. The goal is to proceed and calculate $F_s^k$.
    \item {\bf Candidate Sets Generation:} Each $P_m$ generates a set of candidate $k$-itemsets $B_s^{k,m}$ out of $F_s^{k-1,m} \cap F_s^{k-1}$ --- the
    $(k-1)$-itemsets that are both globally and locally frequent, using the Apriori algorithm.
    \item {\bf Local Pruning:} For each $X \in B_s^{k,m}$, $P_m$ computes $supp_m(X)$ and
    retains only those itemsets that are locally $s$-frequent. We denote  this collection of itemsets
    by $C_s^{k,m}$.
    \item {\bf Unifying the candidate itemsets:} Each player broadcasts his $C_s^{k,m}$ and then all players
    compute $C_s^k:=\bigcup_{m=1}^{M} C_s^{k,m}$.
    \item {\bf Computing local supports.} All players compute the local supports of all itemsets in $C_s^{k}$.
    \item {\bf Broadcast Mining Results:} Each player broadcasts the local supports that he computed. From that, everyone can compute the global support of every
    itemset in  $C_s^k$. Finally, $F_s^k$ is the subset of $C_s^k$ that consists of all globally $s$-frequent $k$-itemsets.
\end{enumerate}

\comment{@@@
\noindent
(When $k=1$, $C_s^{1,m}=F_s^{1,m}$ --- the set of single items that are $s$-frequent in $D_m$.)
}

\subsection{Overview and organization of the paper}
The FDM protocol violates privacy in two stages: In Stage 4, where the
players broadcast the itemsets that are locally frequent
in their private databases, and in Stage 6, where they broadcast the sizes of the local supports of candidate itemsets. Kantarcioglu and Clifton \cite{KanCli}
proposed secure implementations of those two stages. Our improvement is with regard to the secure implementation of Stage 4, which is the more costly stage
of the protocol, and the one in which the protocol of \cite{KanCli} leaks excess information.
In Section \ref{ss1} we describe \cite{KanCli}'s secure implementation of Stage 4. We then describe our alternative implementation
and we proceed to compare the two implementations in terms of privacy of efficiency.
In the next two short Sections \ref{ss2} and \ref{ss3} we describe briefly, for the sake of completeness, \cite{KanCli}'s implementation
of the two remaining stages of the distributed protocol: The identification of those candidate itemsets that are globally $s$-frequent,
and then the derivation of all $(s,c)$-association rules.
Section \ref{rel} includes
a review of related work. We conclude the paper in Section \ref{conc}.

Like in \cite{KanCli} we assume that the players are semi-honest; namely, they follow the protocol but try to extract as much information as possible from their own view.
(See \cite{JiC,SWG04,ZYW} for a discussion and justification of that assumption.)
We too, like \cite{KanCli}, assume that $M>2$. (The case $M=2$ is discussed in \cite[Section 5]{KanCli}; the conclusion is that the
problem of secure computation of frequent itemsets and association rules in the two-party case is unlikely to be of use.)

\section{Secure computation of all locally frequent itemsets}\label{ss1}
Here we discuss the secure computation of the union $C_s^k=\bigcup_{m=1}^{M} C_s^{k,m}$. We describe the protocol of \cite{KanCli} (Section \ref{ss1-KC})
and then our protocol (Sections \ref{or}--\ref{ss1-ours}). We analyze the privacy of the two protocols
in Section \ref{ss1-pri}, their communication cost in Section \ref{ss1-eff}, and their computational cost in Section \ref{ss1-eff1}.

\subsection{The protocol of \cite{KanCli} for the secure computation of all locally frequent itemsets}\label{ss1-KC}
Protocol \ref{alg:1} (\textsc{UniFI-KC} hereinafter) is the protocol that was suggested by Kantarcioglu and Clifton \cite{KanCli}
for computing the unified list of all locally frequent itemsets, $C_s^k=\bigcup_{m=1}^{M} C_s^{k,m}$, without disclosing the sizes
of the subsets $C_s^{k,m}$ nor their contents.
It is based on two ideas: Hiding the sizes of the subsets $C_s^{k,m}$ by means of fake itemsets, and hiding their content by means of encryption.
Let $Ap(F_s^{k-1})$ denote the set of $k$-itemsets which the Apriori algorithm generates when applied on $F_s^{k-1}$.
Clearly, $C_s^{k,m} \subseteq Ap(F_s^{k-1})$, whence $|C_s^{k,m}| \leq |Ap(F_s^{k-1})|$, for all $1 \leq m \leq M$.
Therefore, after each player computes $C_s^{k,m}$, he adds to it fake itemsets until its size becomes $|Ap(F_s^{k-1})|$ in order to
hide the number of locally frequent itemsets that he has. In order to hide the actual itemsets, they use
a commutative encryption algorithm.
(A commutative encryption means that $ E_{K_1} \circ E_{K_2} = E_{K_2} \circ E_{K_1}$ for any pair of keys $K_1$ and $K_2$.)
We proceed to describe the protocol. Since all protocols that we present here
involve cyclic communication rounds, $P_{M+1}$ always means $P_1$, while $P_0$ means $P_M$.

\begin{algorithm}[t!!!]
\begin{algorithmic}[1]
\caption{\label{alg:1} (\textsc{UniFI-KC}) Unifying lists of locally Frequent Itemsets \cite{KanCli}}

\INPUT Each player $P_m$ has an input set $C_s^{k,m} \subseteq Ap(F_s^{k-1})$, $1 \leq m \leq M$.
\OUTPUT $C_s^k=\bigcup_{m=1}^{M} C_s^{k,m} $.

%\vskip 0.2cm
\STATE \textit{\bf \underline{Phase 0:} Getting started}
\STATE The players decide on a commutative cipher and each player $P_m$, $1 \leq m \leq M$, selects a random secret encryption key $K_m$.
\STATE The players select a hash function $h$ and compute $h(x)$ for all $x \in Ap(F_s^{k-1})$.
\STATE If there exist $x_1 \neq x_2 \in Ap(F_s^{k-1})$ for which $h(x_1) = h(x_2)$, select a different $h$.
\STATE Build a lookup table $T=\{(x,h(x)): x \in Ap(F_s^{k-1})\}$.

%\vskip 0.2cm
\STATE \textit{\bf \underline{Phase 1:} Encryption of all itemsets}
\FORALL {Player $P_m$, $1 \leq  m \leq M$,}
\STATE Set $X_m=\emptyset$.
\FORALL {$x \in C_s^{k,m}$}
\STATE Player $P_m$ computes $E_{K_m}(h(x))$ and adds it to $X_m$.
\ENDFOR
\STATE Player $P_m$ adds to $X_m$ faked itemsets until its size becomes $|Ap(F_s^{k-1})|$.
\ENDFOR
\FOR {$i = 2$ to  $M$}
\STATE $P_m$ sends a permutation of $X_m$ to $P_{m+1}$.
\STATE $P_m$ receives from $P_{m-1}$ the permuted $X_{m-1}$.
\STATE $P_m$ computes a new $X_m$ as the encryption of the permuted $X_{m-1}$ using $K_m$.
\ENDFOR

%\vskip 0.3cm

\STATE \textit{\bf \underline{Phase 2:} Merging itemsets}
\STATE Each odd player sends his encrypted sets to player $P_1$.
\STATE Each even player sends his encrypted sets to player $P_2$.
\STATE $P_1$ unifies all sets that were sent by the odd players and removes duplicates.
\STATE $P_2$ unifies all sets that were sent by the even players and removes duplicates.
\STATE $P_2$ sends his permuted list of itemsets to $P_1$.
\STATE $P_1$ unifies his list of itemsets and the list received from $P_2$ and then removes duplicates from the unified list.
Denote the final list by $EC_s^k$.

%\vskip 0.3cm

\STATE \textit{\bf \underline{Phase 3:} Decryption}
\FOR{$m$ = 1 to $M-1$}
\STATE $P_m$ decrypts all itemsets in $EC_s^k$ using $K_m$.
\STATE $P_m$ sends the permuted (and $K_m$-decrypted) $EC_s^k$ to $P_{m+1}$.
\ENDFOR
\STATE $P_{M}$ decrypts all itemsets in $EC_s^k$ using $E_{M}$; denote the resulting set by $C_s^k$.
\STATE $P_{M}$ uses the lookup table $T$ to remove from $C_s^k$ faked itemsets.
\STATE $P_{M}$ broadcasts $C_s^k$.

\end{algorithmic}
\end{algorithm}

In the preliminary Phase 0 (Steps 2-5) the players select the needed cryptographic primitives: They jointly select a commutative cipher, and each player selects
a corresponding random private key.  In addition, they select a hash function $h$ to apply on all itemsets prior to encryption.
It is essential that $h$ will not experience collusions
on $Ap(F_s^{k-1})$ in order to make it invertible on $Ap(F_s^{k-1})$.
Hence, if such collusions occur (an event of a very small probability), a different hash function must be selected.
At the end, the players compute a lookup table with the hash values of all candidate itemsets in $Ap(F_s^{k-1})$, which will be used later on to find the
preimage of a given hash value.

In Phase 1, all players compute a composite encryption of the hashed sets $C_s^{k,m}$, $1 \leq m \leq M$.
First (Steps 7-13), each player $P_m$ hashes all itemsets in $C_s^{k,m}$ and then encrypts them using the key $K_m$. (Hashing is needed in order
to prevent leakage of algebraic relations between itemsets, see \cite[Appendix]{KanCli}.)
Then, he adds to the resulting set faked itemsets until its size becomes
$|Ap(F_s^{k-1})|$. We denote the resulting set by $X_m$. Then (Steps 14-18), the players start a loop of $M-1$ cycles, where in each cycle they perform
the following operation: Player $P_m$ send a permutation of
$X_m$ to the next player $P_{m+1}$; Player $P_m$ receives from $P_{m-1}$ a permutation of the set $X_{m-1}$ and then computes a
new $X_m$ as $X_m=E_{K_m}(X_{m-1})$. At the end of this loop, $P_m$ holds an encryption of the hashed $C_s^{k,m+1}$ using all $M$ keys. Due
to the commutative property of the selected cipher, Player $P_m$ holds the set
$ \{ E_{M}(\cdots (E_2(E_1(h(x)))) \cdots) : x \in C_s^{k,m+1} \}$.

In Phase 2 (Steps 20-25), all odd players send the sets that they have to $P_1$, who unifies them. Consequently,
$P_1$ will hold all hashed and encrypted itemsets of the even players.
Similarly, all even players send the sets that they have to $P_2$,
who also unifies them; hence, $P_2$ will hold all hashed and encrypted itemsets of the odd players.
At this stage, both $P_1$ and $P_2$ remove duplicates, as those
stand (with high probability) for itemsets that were frequent in more than one site.
Then, player $P_2$ permutes his list and sends it to $P_1$ who unifies it with the list that he got. Therefore, at the completion of this stage $P_1$ holds
the union set $C_s^k = \bigcup_{m=1}^{M} C_s^{k,m}$ hashed and then
encrypted by all encryption keys, together with some fake itemsets that were used for the sake of hiding
the sizes of the sets $C_s^{k,m}$; those fake itemsets are not needed anymore and will be removed after decryption in the next
phase.

In phase 3 (Steps 27-33) a similar round of decryptions is initiated. At the end, the last player who performs the last decryption
uses the lookup table $T$ that was constructed in Step 5 in order to identify and remove the fake itemsets and then
to recover $C_s^k$. Finally, he broadcasts $C_s^k$ to all his peers.

\subsection{A secure multiparty protocol for computing the \texttt{OR}}\label{or}
Protocol \textsc{UniFI-KC} securely computes of the union of private subsets of some publicly known ground set ($Ap(F_s^{k-1})$).
Such a problem is equivalent to the problem of computing the \texttt{OR} of private vectors.
Indeed, if we let $n$ denote the size of the ground set, then the private subset that
player $P_m$, $1 \leq m \leq M$, holds may be described by a binary vector $\bb_m \in \ZZ_2^n$, and the union of the private subsets is
described by the \texttt{OR} of those private vectors,
$\bb:= \bigvee_{m=1}^{M} \bb_m$.
%Such a simple function can be evaluated securely by the generic solutions suggested in
%\cite{BMR,BNP,GMW87}.
%We present here two protocols for computing that function which are much simpler to understand and program, and employ less cryptographic primitives. The two protocols assume that the players are semi-honest; namely, they follow the protocol but try to extract as much information as possible from their own view. (In particular, we assume that the players do not collude.)
We present here a protocol for computing that function which
is simpler than \textsc{UniFI-KC} and employs less cryptographic primitives.

The protocol that we present (Protocol \ref{algorithm:thresh2})
computes a wider range of functions, which we call threshold functions.

\begin{definition}
Let $b_1,\ldots,b_M$ be $M$ bits and $1 \leq t \leq M$ be an integer. Then
\be T_t(b_1,\ldots,b_M) = \left\{ \begin{array}{ll} 1 & \mbox{~~~if~~} \sum_{m=1}^M b_m \geq t \\ ~&~ \\ 0 & \mbox{~~~if~~} \sum_{m=1}^M b_m < t \end{array} \right. \label{thres}\ee
is called the $t$-threshold function.
Given binary vectors $\bb_m=(b_m(1),\ldots,b_m(n))  \in \ZZ_2^n$, we let
$ T_t(\bb_1,\ldots,\bb_M)$ denote the binary vector in which the $i$th component equals $T_t(b_1(i),\ldots,b_M(i))$, $1 \leq i \leq n$.
\end{definition}

\noindent
The \texttt{OR} and \texttt{AND} functions are the $1$- and $M$-threshold functions, respectively; i.e.,
$$ \bigvee_{m=1}^{M} \bb_m = T_1(\bb_1,\ldots,\bb_M) \,, \quad \mbox{and} \quad  \bigwedge_{m=1}^{M} \bb_m = T_M(\bb_1,\ldots,\bb_M)\,.$$
Those special cases may be used, as we show in Section \ref{ss1-ours},
to compute in a privacy-preserving manner unions and intersections of subsets.

The main idea behind Protocol \ref{algorithm:thresh2} (\textsc{Threshold} henceforth), which is based on the protocol
suggested in \cite{Ben} for secure computation of the sum, is to compute shares of the sum vector
and then use those shares to securely verify the threshold conditions in each component. Since the sum vector may be seen as a vector over $\ZZ_{M+1}$,
each player starts by creating random shares in $\ZZ_{M+1}^n$ of his input vector (Step 1).
In Step 2 all players send to all other players the corresponding shares in their input vector.
Then (Step 3), player $P_\ell$, $1 \leq \ell \leq M$, adds the shares that he got and arrives at his share, $\bs_\ell$, in the sum vector.
Namely, if we let $\ba: = \sum_{m=1}^M \bb_m$ denote the sum of the input vectors, then
$ \ba= \sum_{\ell=1}^M \bs_\ell \mod (M+1) $;
furthermore, any $M-1$ vectors out of $\{\bs_1,\ldots,\bs_M\}$ do not reveal any information on the sum $\ba$.
In Steps 4-5, all players, apart from the last one, send their shares to $P_1$ who adds them up to get the share $\bs$.
Now, players $P_1$ and $P_M$ hold additive shares of the sum vector $\ba$: $P_1$ has $\bs$, $P_M$ has $\bs_M$, and $\ba = (\bs+\bs_M)$ mod $(M+1)$.
It is now needed to check for each component $1 \leq i \leq n$ whether
$(s(i)+s_M(i)) \mbox{~mod~} (M+1) <t$. Equivalently, we need to check whether
\be (s(i)+s_M(i)) \mbox{~mod~} (M+1) \in \{j: 0 \leq j \leq t-1\} \,. \label{yuy}\ee
The inclusion in (\ref{yuy}) is equivalent to
\be s(i) \in \Theta(i):=\{ (j-s_M(i)) \mbox{~mod~} (M+1): 0 \leq j \leq t-1\} \,. \label{yuy1}\ee
The value of $s(i)$ is known only to $P_1$ while the set $\Theta(i)$ is known only to $P_M$.
The problem of verifying the set inclusion in Eq. (\ref{yuy1})
can be seen as a simplified version of the \emph{privacy-preserving keyword search}, which was solved by Freedman et. al. \cite{ks}.
In the case of the the \texttt{OR} function, $t=1$, which is the relevant case for us, the set $\Theta(i)$ is of size 1, and therefore it is the problem
of oblivious string comparison, a problem that was solved in e.g. \cite{FNW}.
However, we claim that, since $M>2$, there is no need to invoke neither of the secure protocols of \cite{ks} or \cite{FNW}.
Indeed, as $M>2$, the existence of
other semi-honest players can be used to verify the inclusion in Eq. (\ref{yuy1}) much more easily.
This is done in Protocol \ref{algorithm:compare} (\textsc{SetInc}) which we proceed to describe next.

\begin{algorithm}[h!!]
\caption{\label{algorithm:thresh2} (\textsc{Threshold}) Secure computation of the $t$-threshold function}
\begin{algorithmic}[1]
\INPUT Each player $P_m$ has an input binary vector $\bb_m  \in \ZZ_2^n$, $1 \leq m \leq M$.
\OUTPUT $\bb := T_t(\bb_1,\ldots,\bb_M) $.
\STATE Each $P_m$ selects $M$ random share vectors $\bb_{m,\ell} \in \ZZ_{M+1}^n$, $1 \leq \ell \leq M$, such that $\sum_{\ell=1}^M \bb_{m,\ell}=\bb_m$ mod $(M+1)$.
\STATE Each $P_m$ sends $\bb_{m,\ell}$ to $P_\ell$ for all $1 \leq \ell \neq m \leq M$.
\STATE Each $P_\ell$ computes $\bs_\ell =(s_\ell(1),\ldots,s_\ell(n)):= \sum_{m=1}^M \bb_{m,\ell}$ mod $(M+1)$.
\STATE Players $P_\ell$, $2 \leq \ell \leq M-1$, send $\bs_\ell$ to $P_1$.
\STATE $P_1$ computes $\bs =(s(1),\ldots,s(n)):= \sum_{\ell=1}^{M-1} \bs_{\ell}$ mod $(M+1)$.
\FOR {$i=1,\ldots,n$}
\STATE If $(s(i) + s_M(i) ) \mod (M+1) < t$ set $b(i)=0$ otherwise set $b(i)=1$.
\ENDFOR
\STATE Output $\bb=(b(1),\ldots,b(n))$.
\end{algorithmic}
\end{algorithm}

Protocol \textsc{SetInc} starts with players $P_1$ and $P_M$ agreeing
on a keyed hash function $h_K(\cdot)$ (e.g., HMAC \cite{hmac}), a corresponding secret key $K$, and a long random bit string $r$ (Steps 1-2).
Then (Steps 3-4), $P_1$ converts his sequence of elements $\bs=(s(1),\ldots,s(n))$ into a sequence of corresponding ``signatures"
$\bs'=(s'(1),\ldots,s'(n))$, where $s'(i)=h_K(r,i, s(i))$ and $P_M$ does a similar conversions to the subsets that he holds.
$P_1$ then sends $\bs'$ to $P_2$ and $P_M$ sends to $P_2$ a random permutation of each of the subsets $\Theta'(i)$, $1 \leq i \leq n$.
Finally, $P_2$ performs the relevant inclusion verifications on the signature values.
If he finds out that for a given $1 \leq i \leq n$, $s'(i) \in \Theta'(i)$, he may infer, with high probability, that $s(i) \in \Theta(i)$, whence
he sets $b(i)=0$. If, on the other hand, $s'(i) \notin \Theta'(i)$, then, with certainty, $s(i) \notin \Theta(i)$, whence
he sets $b(i)=1$.

Two comments are in order:

\begin{enumerate}[(1)]
\item
If the index $i$ had not been part of the input to the hash function (Steps 3-4), then two equal components in $P_1$'s input vector, say
$s(i)=s(j)$, would have been mapped to two equal signatures, $s'(i)=s'(j)$.
Hence, in that case player $P_2$ would have learnt that in $P_1$'s input vector the $i$th and $j$th
components are equal. To prevent such leakage of information, we include the index $i$ in the input to the hash function.
\item
An event in which $s'(i) \in \Theta'(i)$ while $s(i) \notin \Theta(i)$ indicates a collusion; specifically, it implies that there exist $\theta' \in \Theta(i)$
and $\theta'' \in \Omega \setminus \Theta(i)$ for which $h_K(r,i,\theta')=h_K(r,i,\theta'')$. Hash functions are designed so that the probability of such collusions
is negligible, whence the risk of a collusion can be ignored. However, it is possible for player $P_M$ to check upfront the selected random
values ($K$ and $r$)
in order to verify that for all $1 \leq i \leq n$, the sets
$\Theta'(i) = \{ h_K(r,i,\theta) : \theta \in \Theta(i)\}$ and
$\Theta''(i) = \{ h_K(r,i,\theta) : \theta \in \Omega \setminus \Theta(i)\}$ are disjoint.
\end{enumerate}

\begin{algorithm}[h!!]
\caption{\label{algorithm:compare} (\textsc{SetInc}) Set Inclusion computation}
\begin{algorithmic}[1]
\INPUT $P_1$ has a vector $\bs=(s(1),\ldots,s(n))$ and $P_M$ has a vector $\bTheta=(\Theta(1),\ldots,\Theta(n))$, where for all $1 \leq i \leq n$,
$s(i) \in \Omega$ and $\Theta(i) \subseteq \Omega$ for some ground set $\Omega$.
\OUTPUT The vector $\bb=(b(1),\ldots,b(n))$ where $b(i)=0$ if $s(i) \in \Theta(i)$ and $b(i)=1$ otherwise, $1 \leq i \leq n$.
\STATE $P_1$ and $P_M$ agree on a keyed-hash function $h_K(\cdot)$ and on a secret key $K$.
\STATE $P_1$ and $P_M$ agree on a long random bit string $r$.
\STATE $P_1$ computes $\bs'=(s'(1),\ldots,s'(n))$, where $s'(i)=h_K(r,i,s(i))$, $1 \leq i \leq n$.
\STATE $P_M$ computes $\bTheta'=(\Theta'(1),\ldots,\Theta'(n))$, where $\Theta'(i) = \{ h_K(r,i,\theta) : \theta \in \Theta(i)\}$, $1 \leq i \leq n$.
\STATE $P_1$ sends to $P_2$ the vector $\bs'$.
\STATE $P_M$ sends to $P_2$ the vector $\bTheta'$ in which each $\Theta(i)$ is randomly permuted.
\STATE For all $1 \leq i \leq n$, $P_2$ sets $b(i)=0$ if $s'(i) \in \Theta'(i)$ and $b(i)=1$ otherwise.
\STATE $P_2$ broadcasts the vector $\bb=(b(1),\ldots,b(n))$.
\end{algorithmic}
\end{algorithm}

We refer hereinafter to the combination of Protocols \textsc{Threshold} and \textsc{SetInc} as Protocol
\textsc{Threshold-C}; namely, it is
Protocol \textsc{Threshold} where the inequality verifications in Steps 6-8 are carried out by Protocol \textsc{SetInc}. Then our claims are as follows:

\begin{theorem}\label{prop2} Assume that the $M>2$ players are semi-honest and that
the keyed hash function in Protocol \textsc{SetInc} is preimage-resistant. Then:
\begin{enumerate}[(a)]
\item Protocol \textsc{Threshold-C} is correct (i.e., it computes the threshold function).
\item Let $C \subset \{P_1,P_2,\ldots,P_M\}$ be a coalition of players.
\begin{enumerate}[(i)]
\item
If $P_2 \notin C$ and at least one of $P_1$ and $P_M$ is not in $C$ either, then Protocol \textsc{Threshold-C} is perfectly private with respect to $C$.
\item
If $P_2 \in C$ but $P_1,P_M \notin C$, the protocol is computationally private with respect to $C$.
\item
Otherwise, $C$ includes at least two of the players $P_1,P_2,P_M$; such coalitions may learn the sum $\ba = \sum_{m=1}^M \bb_m$, but no further information beyond the sum.
\end{enumerate}
\end{enumerate}
\end{theorem}

{\em Proof.}

(a) Protocol \textsc{Threshold} operates correctly if the inequality verifications in Step 7 are carried out correctly, since $(s(i) + s_M(i) ) \mod (M+1)$ equals
the $i$th component $a(i)$ in the sum vector $\ba = \sum_{m=1}^M \bb_m$. The inequality verification is correct if Protocol \textsc{SetInc} is correct.
The latter protocol is indeed correct if the randomly selected $K$ and $r$ are such that
for all $1 \leq i \leq n$, the sets
$\Theta'(i) = \{ h_K(r,i,\theta) : \theta \in \Theta(i)\}$ and
$\Theta''(i) = \{ h_K(r,i,\theta) : \theta \in \Omega \setminus \Theta(i)\}$ are disjoint. (As discussed earlier, such a verification can be carried out
upfront, and most all selections of $K$ and $r$ are expected to pass that test.)

(b) Any single player $P_\ell$, $1 \leq \ell \leq M$, learns in the course of the protocol his share $\bs_\ell$ of the sum $\ba$ in a $M$-out-of-$M$ secret sharing
scheme for $\ba$ (see Step 3 in Protocol \textsc{Threshold}). Two players learn more information: $P_1$ receives the shares $\bs_2,\ldots,\bs_{M-1}$
(Step 4 in Protocol \textsc{Threshold}) and $P_2$ receives the signatures $\bs'$ and $\bTheta'$ during Protocol \textsc{SetInc}.
\begin{enumerate}[(i)]
\item If $P_2,P_1 \notin C$ then the players in $C$ have, at most, the shares $\bs_3,\ldots,\bs_M$. Since the secret sharing scheme is perfect, any number of shares
which is less than $M$ reveals no information on $\ba$, since the missing shares were chosen at random.
If $P_2,P_M \notin C$, then the worst scenario is that in which $P_1 \in C$; in that case,
the coalition knows the shares $\bs_1,\ldots,\bs_{M-1}$. Once again, as the missing share $\bs_M$ was chosen at random by $P_M$, the shares
$\bs_1,\ldots,\bs_{M-1}$ reveal no information
on $\ba$.
\item If $P_2 \in C$ and $P_1,P_M \notin C$, then $C$ has at most the $M-2$ additive shares $\bs_2,\ldots,\bs_{M-1}$.
The additional knowledge that $P_2$ holds enables, in theory, the recovery of the missing shares $\bs_1$ and $\bs_M$ and then the recovery of
$\ba=\sum_{\ell=1}^M \bs_\ell$.
Indeed, by scanning all possible keys $K$ of the keyed hash function and all possible random strings $r$, $P_2$
may find a key $K$ and a string $r$ for which
the signature values that he got from $P_1$ and $P_M$ (namely, $\bs'$ and $\bTheta'$) are consistent with the signature scheme and the elements of $\Omega=\{0,1,\ldots,M\}$.
Hence, the protocol does not provide perfect privacy in the information-theoretic sense
with respect to such coalitions. However, since such a computation is infeasible, and as the hash function is preimage-resistant,
the protocol provides computational privacy with respect to such coalitions.
\item If $P_1,P_M \in C$, then by adding $\bs$ (known to $P_1$) and $\bs_M$ (known to $P_M$), they will get the sum $\ba$. No further information on the input vectors
$\bb_1,\ldots,\bb_M$
may be deduced from the inputs of the players in such a coalition;
specifically, every set of vectors $\bb_1,\ldots,\bb_M$ that is consistent
with the sum $\ba$ is equally likely.
Coalitions $C$ that include either $P_1,P_2$ or $P_2,P_M$ can also recover $\ba$.
Indeed, $P_2$ knows $\bs'$ and $\bTheta'$ and $P_1$ or $P_M$ knows $h_K$, $K$ and $r$.
Hence, a coalition of $P_2$ with either $P_1$ or $P_M$ may recover from those values the preimages $\bs$ and $\bTheta$. Hence, such
a coalition can recover $\bs$ and $\bs_M$, and consequently $\ba$. As argued before, the shares available for such coalitions do not reveal any further information
about the input vectors $\bb_1,\ldots,\bb_M$.
\end{enumerate}
$\Box$

The susceptibility of Protocol \textsc{Threshold-C} to coalitions is not
very significant because of two reasons:

\begin{itemize}
\item The entries of the sum vector $\ba$ do not reveal
information about specific input vectors. Namely, knowing that $a(i)=p$ only indicates that $p$ out of the $M$ bits $b_m(i)$, $1 \leq m \leq M$, equal 1, but it
reveals no information regarding which of the $M$ bits are those.
\item There are only three players that can collude in order to learn information beyond the intention of the protocol.
Such a situation is far less severe than a situation in which
any player may participate in a coalition,
since if it is revealed that a collusion took place, there is a small set of suspects.
\end{itemize}

\subsection{An improved protocol for the secure computation of all locally frequent itemsets}\label{ss1-ours}
As before, we denote by $F_s^{k-1}$ the set of all globally frequent $(k-1)$-itemsets, and by $Ap(F_s^{k-1})$
the set of $k$-itemsets that the Apriori algorithm generates when applied on $F_s^{k-1}$.
All players can compute that set and decide on an ordering of it. (Since all itemsets are subsets of $A=\{a_1,\ldots,a_L\}$, they may be viewed
as binary vectors in $\{0,1\}^L$ and, as such, they may be ordered lexicographically.)
Then, since the sets of locally frequent $k$-itemsets, $C_s^{k,m}$, $1 \leq m \leq M$,
are subsets of $Ap(F_s^{k-1})$, they may be encoded as binary vectors of length $n_k:=|Ap(F_s^{k-1})|$.
The binary vector that encodes the union $C_s^k:=\bigcup_{m=1}^{M} C_s^{k,m}$ is the \texttt{OR} of the vectors that encode the
sets $C_s^{k,m}$, $1 \leq m \leq M$.
Hence, the players can compute the union by invoking Protocol \textsc{Threshold-C} on their binary input vectors. This approach is summarized in Protocol
\ref{algorithm:union} (\textsc{UniFI}). (Replacing $T_1$ with $T_M$ in Step 2 will result in computing the intersection of the private subsets.)

\begin{algorithm}[h!!]
\caption{\label{algorithm:union} (\textsc{UniFI}) Unifying lists of locally Frequent Itemsets}
\begin{algorithmic}[1]
\INPUT Each player $P_m$ has an input subset $C_s^{k,m} \subseteq Ap(F_s^{k-1})$, $1 \leq m \leq M$.
\OUTPUT $C_s^k=\bigcup_{m=1}^{M} C_s^{k,m} $.
\STATE Each player $P_m$ encodes his subset $C_s^{k,m}$ as a binary vector $\bb_m$ of length $n_k=|Ap(F_s^{k-1})|$, in accord with the agreed ordering of $Ap(F_s^{k-1})$.
\STATE The players invoke Protocol \textsc{Threshold-C} to compute $\bb=T_1(\bb_1,\ldots,\bb_M)=\bigvee_{m=1}^M \bb_m$.
\STATE $C_s^k$ is the subset of $Ap(F_s^{k-1})$ that is described by $\bb$.
\end{algorithmic}
\end{algorithm}

\subsection{Privacy}\label{ss1-pri}
We begin by analyzing the privacy offered by Protocol \textsc{UniFI-KC}.
That protocol does not respect perfect privacy since it reveals to the players information that is not implied by their own input and the final output.
In Step 11 of Phase 1 of the protocol, each player augments the set $X_m$ by fake itemsets. To avoid unnecessary hash and encryption computations, those
fake itemsets are random strings in the ciphertext domain of the chosen commutative cipher.
The probability of two players selecting random strings that will become equal at the end of Phase 1 is negligible; so is the probability of Player $P_m$
to select a random string that equals $E_{K_m}(h(x))$ for a true itemset $x \in Ap(F_s^{k-1})$. Hence, every encrypted itemset that appears in two different lists indicates
with high probability a true itemset that is locally $s$-frequent in both of the corresponding sites. Therefore,
Protocol \textsc{UniFI-KC} reveals the following excess information:
\begin{enumerate}[(1)]
\item $P_1$ may deduce for any subset of the even players, the number of itemsets that are locally supported by all of them.
\item $P_2$ may deduce for any subset of the odd players, the number of itemsets that are locally supported by all of them.
\item $P_1$ may deduce the number of itemsets that are supported by at least one odd player and at least one even player.
\item If $P_1$ and $P_2$ collude, they reveal for any subset of the players the number of itemsets that are locally supported by all of them.
\end{enumerate}

As for the privacy offered by Protocol \textsc{UniFI}, we consider two cases:
If there are no collusions, then, by Theorem \ref{prop2},
Protocol \textsc{UniFI} offers perfect privacy with respect to
all players $P_m$, $m \neq 2$,
and computational privacy with respect to $P_2$.
This is a privacy guarantee better than that offered by
Protocol \textsc{UniFI-KC}, since the latter protocol does reveal information to $P_1$ and $P_2$ even if they do not collude with any other player.

If there are collusions, both Protocols \textsc{UniFI-KC} and \textsc{UniFI} allow the colluding parties to learn forbidden information.
In both cases, the number of ``suspects" is small --- in Protocol \textsc{UniFI-KC} only $P_1$ and $P_2$ may benefit from a collusion while in Protocol \textsc{UniFI}
only $P_1$, $P_2$ and $P_M$ can extract additional information if two of them collude (see Theorem \ref{prop2}).
In Protocol \textsc{UniFI-KC}, the excess information which may be extracted
by $P_1$ and $P_2$ is about the number of common frequent itemsets among any subset of the players. Namely, they may learn that, say, $P_2$ and $P_3$ have
many itemsets that are frequent in both of their databases (but not which itemsets), while $P_2$ and $P_4$ have very few itemsets that are frequent in their corresponding
databases. The excess information in Protocol \textsc{UniFI} is different: If any two out of $P_1$, $P_2$ and $P_M$ collude, they can learn
the sum of all private vectors. That sum reveals for each specific
itemset in $Ap(F_s^{k-1})$ the number of sites in which it is frequent, but not which sites.
Hence, while the colluding players in Protocol \textsc{UniFI-KC} can distinguish between the different players and learn about the similarity or dissimilarity between them,
Protocol \textsc{UniFI} leaves the partial databases totally indistinguishable, as the excess information that it leaks
is with regard to the itemsets only.

To summarize, given that Protocol \textsc{UniFI} reveals no excess information when there are no collusions, and, in addition,
when there are collusions, the excess information still leaves the partial databases indistinguishable, it offers enhanced privacy preservation in comparison to
Protocol \textsc{UniFI-KC}.

\subsection{Communication cost}\label{ss1-eff}
Here and in the next section we analyze the communication and computational costs of Protocols \textsc{UniFI-KC} and \textsc{UniFI}.
In doing so, we let $n_k:=|Ap(F_s^{k-1})|$ and $n:=\sum_{k=2}^L n_k$; also, the $k$th iteration refers to the iteration
in which $F^k_s$ is computed from $F^{k-1}_s$ ($2 \leq k \leq L$).

We start with Protocol \textsc{UniFI-KC}.
Let $t$ denote the number of bits required to represent an itemset.
Clearly, $t$ must be at least $\log_2 n_k$ for all $2 \leq k \leq L$. However, as Protocol \textsc{UniFI-KC} hashes the itemsets and then encrypts them,
$t$ should be at least the recommended
ciphertext length in commutative ciphers.
RSA \cite{RSA}, Pohlig-Hellman \cite{PH} and ElGamal \cite{ElG} ciphers are examples of commutative encryption schemes. As the
recommended length of the modulus in all of them is at least 1024 bits, we take $t=1024$.

During Phase 1 of Protocol \textsc{UniFI-KC}, there are $M-1$ rounds of communication, where in each one of them
each of the $M$ players sends to the next player a message; the length of that message in the $k$th iteration is $t n_k$.
Hence, the communication cost of this phase in the $k$th iteration is $(M-1)M t n_k$.
During Phase 2 of the protocol all odd players send their encrypted itemsets to $P_1$ and all even players send their encrypted itemsets to $P_2$. Then $P_2$
unifies the itemsets he got and sends them to $P_1$. Hence, this phase takes 2 more rounds and its communication cost in the $k$th iteration is bounded by $1.5M tn_k$.
In the last phase a similar round of decryptions is initiated. The unified list of all encrypted true and fake itemsets may contain in the $k$th iteration
up to
$M n_k$ itemsets. Hence, that phase involves $M-1$ rounds with communication cost of no more than
$(M-1)M t n_k$.
To summarize: Protocol \textsc{UniFI-KC} entails $2M$ communication rounds (in each of the iterations) and the communication cost in the $k$th iteration is $\Theta(M^2 t n_k)$.
(In fact, since the majority of the itemsets are expected to be fake itemsets, the communication cost in the decryption phase is close to $(M-1)M t n_k$
and then the overall communication cost is roughly $2M^2 t n_k$.)

We now turn to analyze the communication cost of Protocol \textsc{UniFI}.
It consists of 3 communication rounds:
One for Step 2 of Protocol \textsc{Threshold} that it invokes, one for Step 4 of that protocol, and one for the threshold verifications in Steps 6-8
(in which Protocol \textsc{SetInc} is invoked).
Hence, it improves upon Protocol \textsc{UniFI-KC} that entails $2M$ communication rounds.

In the $k$th iteration, the length of the vectors in Protocol \textsc{Threshold} is $n_k$;
each entry in those vectors represents a number between 0 and $M-1$, whence it may be encoded by $\log_2 M$ bits.
Therefore:
\begin{itemize}
\item The communication cost of Step 2 in Protocol \textsc{Threshold} is
$M(M-1) (\log_2 M) n_k$ bits.
\item The communication cost of Step 4 in Protocol \textsc{Threshold} is
$(M-2) (\log_2 M) n_k$ bits.
\item  Steps 6-8 in Protocol \textsc{Threshold} are carried out by invoking Protocol \textsc{SetInc}. The communication cost of Steps 5-6 in the latter protocol
is $2|h| n_k$, where $|h|$ is the size in bits of the hash function's output. (Recall that when Protocol \textsc{SetInc} is called in the framework of
Protocol \textsc{Threshold-C}, the size of each of the sets $\Theta(i)$ is 1.)
\end{itemize}
Hence, the overall communication cost of Protocol \textsc{Threshold-C} in the $k$th iteration is $((M^2-2)(\log_2 M) +2|h|) n_k$ bits.

As discussed earlier, a plausible setting of $t$ would be $t=1024$. A typical value of $|h|$ is 160. With those parameters,
the improvement factor in terms of communication cost, as offered by Protocol \textsc{UniFI} with respect to Protocol \textsc{UniFI-KC}, is approximately
$$ \frac{2M^2 tn_k}{(M^2 \log_2 M  + 2|h|)n_k} = \left(\frac{\log_2 M}{2t} + \frac{|h|}{M^2t} \right)^{-1} = \left(\frac{\log_2 M}{2048} + \frac{0.15625}{M^2} \right)^{-1}\,.$$
For $M=4$ we get an improvement factor of roughly 93, while for $M=8$ we get a factor of about 256.

\subsection{Computational cost}\label{ss1-eff1}
In Protocol \textsc{UniFI-KC} each of the players needs to perform hash evaluations as well as encryptions and decryptions. As the cost of hash evaluations is significantly
smaller than the cost of commutative encryption, we focus on the cost of the latter operations.
In Steps 9-11 of the protocol, player $P_m$ performs $|C_s^{k,m}| \leq n_k=|Ap(F_s^{k-1})|$ encryptions (in the $k$th iteration).
Then, in Steps 14-18, each player performs $M-1$ encryptions of sets that include $n_k$ items. Hence, in Phase 1 in the $k$th iteration, each player performs
between $(M-1)n_k$ and $Mn_k$ encryptions.
In Phase 3, each player decrypts the set of items $EC_s^k$. $EC_s^k$ is the union of the encrypted
sets from all $M$ players, where each of those sets has $n_k$ items --- true and fake ones. Clearly, the size of $EC_s^k$ is at least $n_k$. On the other hand, since
most of the items in the $M$ sets are expected to be fake ones, and the probability of collusions between fake items is negligible, it is expected that the size
of $EC_s^k$ would be close to $M n_k$. So, in all its iterations ($2 \leq k \leq L$),
Protocol \textsc{UniFI-KC} requires each player to perform an overall number of
close to $2Mn$ (but no less than $Mn$) encryptions or decryptions, where, as before $n = \sum_{k=2}^L n_k$.
Since commutative encryption is typically based on modular exponentiation, the overall
computational cost of the protocol is $\Theta(M t^3 n)$ bit operations per player.

In Protocol \textsc{Threshold}, which Protocol \textsc{UniFI} calls, each player needs to generate $(M-1)n$ (pseudo)random $(\log_2 M)$-bit numbers (Step 1).
Then, each player performs $(M-1)n$ additions of such numbers in Step 1 as well as in Step 3. Player $P_1$ has to perform also $(M-2)n$ additions in
Step 5.
Therefore, the computational cost for each player is $\Theta(Mn \log_2 M)$ bit operations. In addition,
Players 1 and $M$ need to perform $n$ hash evaluations.

\subsubsection{Estimating the practical gain in computation time}\label{ss1-prac}
Here, we estimate the practical gain in computation time, as offered by Protocol \textsc{UniFI} in comparison to Protocol \textsc{UniFI-KC}.
We adopt the same estimation methodology that was used in \cite[Section 6.2]{KanCli}.
We measured the times to perform the basic operations used by the two protocols on an Intel(R) Core(TM)2 Quad CPU 2.67 GHz personal computer with 8 GB of RAM:
\begin{itemize}
\item
Modular addition took (on average) 0.762 $\mu$s (microseconds);
\item
random byte generation took 0.0126 $\mu$s;
\item
equality verification between two 160-bit values took 0.13 $\mu$s;
\item
modular exponentiation with modulus of $t=1024$ bits
took 12251 $\mu$s; and
\item
computing HMAC on an input of 512 bits took on average 15.7 $\mu$s.
\end{itemize}
As in \cite{KanCli}, we estimate the added computational cost due to the secure computations in the two protocols when implemented in the experimental setting that
was used in
\cite{article2}. In that experimental setting, the number of sites was $M=3$, the number of items was $L=1000$ and the unified database contained $N=500,000$ transactions.
Using a support threshold of $s=0.01$, \cite{article2} reports that $n=\sum_{k=2}^L n_k $ was just over 100,000.

In Protocol \textsc{UniFI-KC} each player performs roughly $2M n$ encryptions and decryptions.
Hence, the corresponding encryption time per player is
$2\cdot 3 \cdot 100,000 \cdot 12251$ $\mu$s in this setting, i.e., approximately 7350 seconds.
In comparison, Protocol \textsc{UniFI} requires each player to generate $(M-1)n \log_2 M$ pseudo random bits (50,000 bytes in this setting, which mean
$630$ $\mu$s), and perform $(2M-2)n$ additions (400,000 modular additions in this case, which mean $0.305$ seconds).
$P_1$ has to perform $(M-2)n$ additional additions ($0.076$ seconds); $P_1$ and $P_M$
need to perform 100,000 HMAC computations (1.57 seconds); and $P_2$ performs 100,000 verifications of equality between 160-bit values ($0.013$ seconds). The overall
computational overhead for the least busy player ($P_2$) is $0.318$ seconds and for the busiest player ($P_1$) is 1.951 seconds. Hence,
compared to 7350 seconds for each player in Protocol \textsc{UniFI-KC}, the improvement in computational time overhead is overwhelming.

\section{Identifying the globally $s$-frequent itemsets}\label{ss2}
Protocol \textsc{UniFI-KC} yields the set $C_s^k$ that consists of all itemsets that are locally $s$-frequent in at least one site. Those are the $k$-itemsets
that have potential to be also globally $s$-frequent. In order to reveal which of those itemsets is globally $s$-frequent there is a need to securely
compute the support of each of those itemsets. That computation must not reveal the local support in any of the sites.
Let $x$ be one of the candidate itemsets in $C_s^k$. Then $x$ is globally $s$-frequent if and only if
\be \Delta(x):=supp(x) -s N=\sum_{m=1}^M (supp_m(x) - s N_m) \geq 0 \,.
\label{sume} \ee
Kantarcioglu and Clifton considered two possible settings.
If the required output includes all globally $s$-frequent itemsets, as well as the sizes of their supports, then the values of $\Delta(x)$ can be revealed for all
$x \in C_s^k$. In such a case, those values may be computed using a secure summation protocol (e.g. \cite{Ben}). The more interesting setting, however, is the one where
the support sizes are not part of the required output. We proceed to discuss it.

Let $q$ be an integer greater than $2N$. Then,
since $|\Delta(x)|\leq N$, the itemset $x \in C_s^k$ is $s$-frequent if and only if $\Delta(x) \mod q < q/2$.
The idea is to verify that inequality by starting an implementation of
the secure summation protocol of \cite{Ben} on the private inputs $\Delta_m(x):=supp_m(x) - s N_m$, modulo $q$.
In that protocol, all players jointly compute random additive shares of the required sum $\Delta(x)$ and then, by sending all shares to, say, $P_1$, he may add them
and reveal the sum. If, however, $P_M$ withholds his share of the sum, then $P_1$ will have
one random share, $s_1(x)$, of $\Delta(x)$, and $P_M$ will have a corresponding share, $s_M(x)$; namely,
$s_1(x)+s_M(x) = \Delta(x)$ mod $q$.
It is then proposed that the two players execute the generic secure circuit evaluation of \cite{Yao} in order to
verify whether
\be ( s_1(x)+s_M(x) ) \mod q \leq q/2\,. \label{ineww}\ee
Those circuit evaluations may be parallelized for all $x \in C_s^k$.

We observe that inequality (\ref{ineww}) holds if and only if
\be s_1(x) \in \Theta(x):=\{ (j - s_M(x)) \mod q: 0 \leq j \leq \lfloor (q-1)/2 \rfloor \} \,. \label{inew2r}\ee
As $s_1(x)$ is known only to $P_1$ while $\Theta(x)$ is known only to $P_M$, the verification of the set inclusion in (\ref{inew2r}) can theoretically
be carried out by means of Protocol
\textsc{SetInc}.
However, the ground set $\Omega$ in this case is $\ZZ_q$, which is typically a large set. (Recall that when Protocol \textsc{SetInc} is invoked from
\textsc{UniFI}, the ground
set $\Omega$ is $\ZZ_{M+1}$, which is typically a small set.) Hence, Protocol
\textsc{SetInc} is not useful in this case, and, consequently, Yao's generic protocol remains, for the moment, the simplest way to securely
verify inequality (\ref{ineww}).
Yao's protocol is designed for the two-party case. In our setting, as $M>2$, there exist additional semi-honest players. (This is the assumption
on which Protocol \textsc{SetInc} relies.)
An interesting question which arises in this context is whether
the existence of such additional semi-honest players may be used to verify inequalities
like (\ref{ineww}), even when the modulus is large, without resorting to costly protocols such as oblivious transfer.

\section{Identifying all $(s,c)$-association rules}\label{ss3}
Once all $s$-frequent itemsets are found ($F_s$), we may proceed to look for all $(s,c)$-association rules (rules with support at least $sN$ and confidence at least $c$).
For $X,Y \in F_s$, the corresponding association rule $X \Rightarrow Y$ has confidence at least $c$ if and only if
$supp(X\cup Y)/supp(X) \geq c$, or, equivalently,
\be C_{X,Y}:= \sum_{m=1}^M \left( supp_m(X \cup Y)-c\cdot supp_m(X) \right) \geq 0\,. \label{eq1s}\ee
Since $|C_{X,Y}| \leq N$, then by taking $q= 2N  +1$, the players can verify inequality (\ref{eq1s}),
in parallel, for all candidate association rules, as described in Section \ref{ss2}.

\section{Related work}\label{rel}
Previous work in privacy preserving data mining has considered two related settings. One, in which the data owner and the data miner are two different entities,
and another, in which the data is distributed among several parties who aim to jointly perform data mining on the unified corpus of data that they hold.

In the first setting, the goal is to protect the data records from the data miner. Hence, the data owner aims at
anonymizing the data prior to its release. The main approach in this context is to apply data perturbation \cite{AS00,Evfi}. The idea is that the perturbed
data can be used to infer general trends in the data, without revealing original record information.

In the second setting, the goal is to perform data mining while
protecting the data records of each of the data owners from the other data owners.
This is a problem of secure multi-party computation. The usual approach here is cryptographic rather than probabilistic.
Lindell and Pinkas \cite{LP00} showed how to securely build an ID3 decision tree when the training set is distributed horizontally.
Lin et al. \cite{LinCZ05} discussed secure clustering using the EM algorithm over horizontally distributed data.
The problem of distributed association rule mining was studied in \cite{VC02,ZMC} in the vertical setting, where each party holds a different set of attributes,
and in \cite{KanCli} in the horizontal setting. Also the work of \cite{SWG04} considered this problem in the horizontal setting, but they considered large-scale
systems in which, on top of the parties that hold the data records (resources) there are also managers which are computers that assist the
resources to decrypt messages; another assumption made in \cite{SWG04} that distinguishes it from \cite{KanCli} and the present study
is that no collusions occur between the
different network nodes --- resources or managers.

\section{Conclusion}\label{conc}
We proposed a protocol for secure mining of association rules in horizontally distributed databases that improves significantly
upon the current leading protocol \cite{KanCli} in terms of privacy and efficiency.
One of the main ingredients in our proposed protocol
is a novel secure multi-party protocol for computing the union (or intersection) of private subsets that each of the interacting players hold.
Another ingredient is a protocol that tests the inclusion of an
element held by one player in a subset held by another. The latter protocol exploits the fact that the underlying problem is of interest only when the number of players
is greater than two.

One research problem that this study suggests was described in
Section \ref{ss2}; namely, to devise an efficient protocol for
set inclusion verification that uses the existence of a semi-honest third party. Such a protocol might enable to
further improve upon the communication and computational costs of
the second and third stages of the protocol of \cite{KanCli}, as
described in Sections \ref{ss2} and \ref{ss3}. Another research
problem that this study suggests is the extension of those
techniques to the problem of mining generalized association rules
\cite{SriAg}.

\bibliographystyle{plain}
\bibliography{DARM}

\end{document}